\documentclass[sigconf,9pt]{acmart}

\AtBeginDocument{%
  \providecommand\BibTeX{{%
    \normalfont B\kern-0.5em{\scshape i\kern-0.25em b}\kern-0.8em\TeX}}}

\setcopyright{acmcopyright}
\copyrightyear{2024}
\acmYear{2024}
\setcopyright{acmlicensed}\acmConference[DAC '24]{61st ACM/IEEE Design Automation Conference}{June 23--27, 2024}{San Francisco, CA, USA}
\acmBooktitle{61st ACM/IEEE Design Automation Conference (DAC '24), June 23--27, 2024, San Francisco, CA, USA}
\acmDOI{10.1145/3649329.3655970}
\acmISBN{979-8-4007-0601-1/24/06}

\usepackage{amsthm}
\usepackage{algorithm}
\usepackage[noend]{algpseudocode}
\usepackage{xcolor}
\usepackage{graphicx}
\usepackage{enumitem}
\usepackage{caption}
\usepackage{mathtools}
\usepackage{subcaption}

\begin{document}

\title{Predicting Lemmas in Generalization of IC3}

\author{Yuheng Su}
\affiliation{
    \institution{University of Chinese Academy of Sciences}
    \institution{Institute of Software, Chinese Academy of Sciences}
    \country{China}
}
\email{gipsyh.icu@gmail.com}

\author{Qiusong Yang}
\authornote{Qiusong Yang is the corresponding author.}
\affiliation{
    \institution{Institute of Software, Chinese Academy of Sciences}
    \country{China}
}
\email{qiusong@iscas.ac.cn}

\author{Yiwei Ci}
\affiliation{
    \institution{Institute of Software, Chinese Academy of Sciences}
    \country{China}
}
\email{yiwei@iscas.ac.cn}

\begin{abstract}
The IC3 algorithm, also known as PDR, has made a significant impact in the field of safety model checking in recent years due to its high efficiency, scalability, and completeness. The most crucial component of IC3 is inductive generalization, which involves dropping variables one by one and is often the most time-consuming step. In this paper, we propose a novel approach to predict a possible minimal lemma before dropping variables by utilizing the counterexample to propagation (CTP). By leveraging this approach, we can avoid dropping variables if predict successfully. The comprehensive evaluation demonstrates a commendable success rate in lemma prediction and a significant performance improvement achieved by our proposed method.
\end{abstract}

\maketitle

\section{Introduction}
As hardware designs continue to scale and grow increasingly complex, the occurrence of errors in circuits becomes more prevalent, and the cost of such errors is substantial. To ensure the correctness of circuits, formal verification techniques such as symbolic model checking have been widely adopted by the industry. Symbolic model checking enables the proof of properties in hardware circuits or the identification of errors that violate these properties. IC3 \cite{IC3} (also known as PDR \cite{PDR}) is a highly influential SAT-based algorithm used in hardware model checking, serving as a primary engine for many state-of-the-art model checkers. In comparison to other symbolic algorithms, IC3 stands out for its completeness compared to BMC and scalability compared to IMC and BDD approaches.

To verify a safety property, IC3 endeavors to compute an inductive invariant. This invariant is derived from a series of frames denoted as $F_0, \cdots, F_k$ , where each $F_i$ represents an over-approximation of the set of states reachable in $i$ or fewer steps and does not include any unsafe states. Each frame is represented as a CNF formula, and each clause within the formula is also referred to as a \textit{lemma}. The process terminates when either $F_i = F_{i+1}$ for some $i$, indicating that an inductive invariant has been reached and the system being verified satisfies the designated safety property, or when a real counterexample is found. The basic algorithm for computing the frame $F_{i+1}$ from the previous $F_i$ first \textit{initializes} the frame $F_{i+1}$  to $\top$, which represents the set of all possible states. It then efficiently \textit{blocks} any unsafe states through inductive generalization and \textit{propagates} the lemmas learned in previous frames to $F_{i+1}$  as far as possible. 

Given an unsafe state represented as a boolean cube, the inductive generalization procedure aims to expand it as much as possible to include additional unreachable states. The standard algorithm \cite{IC3} adopts the "down" strategy by attempting to drop as many literals as possible. If a certain literal is dropped and the negation of the new cube is found to be inductive, this results in a smaller inductive clause. If unsuccessful, the algorithm then attempts the process again with a different literal. This process strengthens the frames and significantly reduces the number of iterations. However, it is also the most time-consuming part and has a significant impact on the performance of IC3. Numerous studies have been undertaken to improve generalizations. According to \cite{BetterGeneralization}, certain states that are unreachable but can be backtracked from unsafe states, which cannot be further generalized, are identified as counterexamples to generalization (CTG). The negation of a CTG proves to be inductive, leading to a reduction in the depth of the explicit backward search performed by IC3 and enabling stronger inductive generalizations. In  \cite{CAV23}, the authors first drop literals that have not appeared in any subsumed lemmas of the previous frame in order to increase the probability of producing lemmas that can be effectively propagated to the next frame. In \cite{PushingToTheTop}, the learned lemmas are labeled as \textit{bad} lemmas if they exclude certain reachable states. Later on, these lemmas are not pushed further and their usage in generalization is minimized.
\begin{figure}[h]
    \centering
    \includegraphics[width=0.45\textwidth]{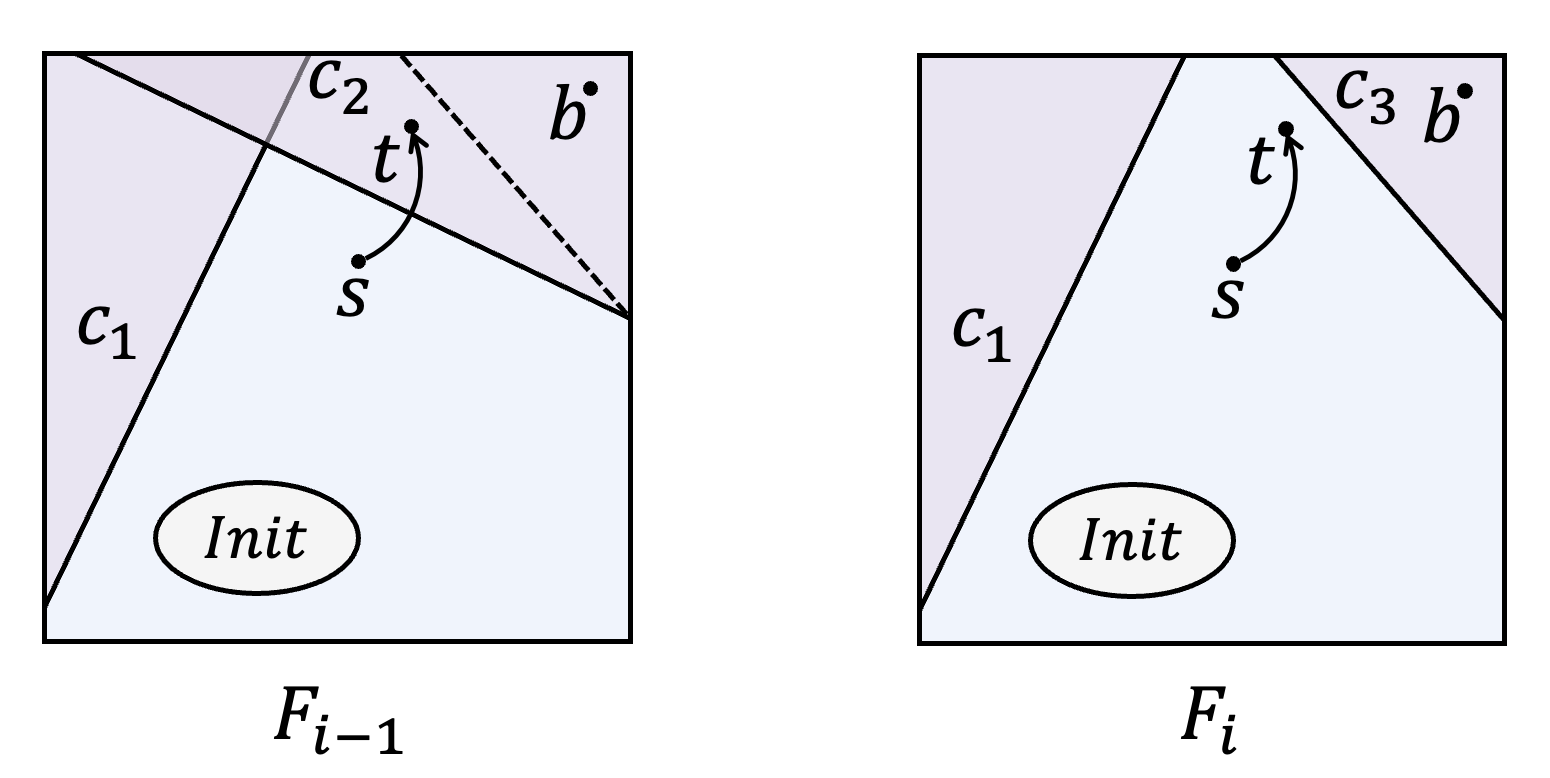}
    \caption{$F_{i-1}$ and $F_i$ frames: the square represents the entire state space, the blue region represents the state space represented by the frame, and each triangle in the purple region represents states blocked by a cube.}
    \label{fig:Fi1Fi}
\end{figure}

In this paper, we aim to address the inefficiency of IC3 generalization in the scenario depicted in Figure \ref{fig:Fi1Fi}. Let's assume there is only one lemma $\neg c_1$ (where $c_1$ represents a cube) in the frame $F_{i-1}$ at first. Later, an unsafe state $b$ is reached, and the generalization of $b$ yields the lemma $\neg c_2$. Furthermore, as the lemmas are further propagated to $F_{i}$, the lemma $\neg c_1$ succeeds while $\neg c_2$ fails. This failure indicates that the cube $c_2$ cannot be blocked in $F_{i}$. In previous studies\cite{IC3,BetterGeneralization,PushingToTheTop,CAV23}, the standard practice was to discard the lemma $\neg c_2$ and then rediscover $b$ in an attempt to block and generalize it once again. If we could predict a potential lemma candidate before dropping variable in generalization based on the information learned so far and verify its validity, we might be able to avoid the costly process of dropping variables one by one. However, the challenge lies in developing a prediction mechanism with a high success rate, as an ineffective prediction would result in additional overhead.
	
The work presented in this paper utilizes counterexample to propagation (CTP) to predict potential lemma candidates and thereby significantly reduce the number of costly generalization operations. As shown in Figure \ref{fig:Fi1Fi}, when the lemma $\neg c_2$ cannot be propagated to $F_{i}$, it indicates the existence of two states, namely $s$ and $t$. Here, $t$ is a successor state of $s$, where $s \models F_{i-1}$ and $t \models c_2$.  We will propose a method for predicting candidates for the lemma $\neg c_3$, which blocks the unsafe state $b$ and is inductive relative to $F_{i-1}$, based on $c_2$, $t$, and $b$. This marks the first implementation of a lemma prediction mechanism in the IC3 algorithm. We conducted a comprehensive performance evaluation of the proposed method, and the results demonstrate that our approach achieves a commendable success rate in predicting lemmas, thereby significantly improving the performance of IC3.

\textit{Organization:} Section \ref{Sec:Preliminaries} provides some background concepts. In Section \ref{Sec:PredictingLemmas}, the algorithm predicting  lemma candidates is presented. The proposed algorithm is comprehensively evaluated in Section \ref{Sec:Evaluation}. Section \ref{Sec:RelatedWork} discusses related works. Section \ref{Sec:Conclusion} concludes this paper and present directions for future work.

\section{Preliminaries}
\label{Sec:Preliminaries}
\subsection{Basics and Notations}
Our context is standard propositional logic. We use notations such as $x, y$ for Boolean variables, and $X, Y$ for sets of Boolean variables. The terms $x$ and $\lnot x$ are referred to as \emph{literals}. \emph{Cube} is conjunction of literals, while \emph{clause} is disjunction of literals. The negation of a cube becomes a clause, and vice versa. A Boolean formula in Conjunctive Normal Form (CNF) is a conjunction of clauses. It is often convenient to treat a clause or a cube as a set of literals, and a CNF as a set of clauses. For instance, given a CNF formula $F$, a clause $c$, and a literal $l$, we write $l \in c$ to indicate that $l$ occurs in $c$, and $c\in F$ to indicate that $c$ occurs in $F$. A formula $F$ implies another formula $G$, if every satisfying assignment of $F$ satisfies $G$, denoted as $F \Rightarrow G$. This relation can also be expressed as $F \models G$.

A Boolean transition system, denoted as $S$, can be defined as a tuple $\langle X, Y, I, T\rangle$. Here, $X$ and $X'$ represent the sets of state variables in the current state and the next state respectively, while $Y$ represents the set of input variables. The Boolean formula $I(X)$ represents the initial states, and $T(X, Y, X')$ describes the transition relation of the system.  State $s_2$ is a successor of state $s_1$ iff $(s_1, s_2')$ is an assignment of $T$ ($(s_1, s_2') \models T$). A $path$ of length $n$ in $S$ is a finite sequence of states $s_0, s_1, \ldots, s_n$ such that $s_0$ is an initial state ($s_0 \models I$) and for each $i$ where $0 \leq i < n$, state $s_{i+1}$ is a successor of state $s_i$ ($(s_i, s_{i+1}') \models T$). Checking a safety property of a transition system $S$ is reducible to checking an invariant property \cite{10.5555/211468}. An invariant property $P(X)$ is a Boolean formula over $X$. A system $S$ satisfies the invariant property $P$ iff all reachable states of $S$ satisfy the property $P$. Otherwise, there exists a counterexample \emph{path} $s_0, s_1, \ldots, s_k$ such that $s_k \not\models P$.

\subsection{Overview of IC3}
IC3 is a SAT-based safety model checking algorithm, which only needs to unroll the system at most once \cite{IC3}. It tries to prove that $S$ satisfies $P$ by finding an inductive invariant $INV(X)$ such that:
\begin{itemize}
\item $I(X) \Rightarrow INV(X)$
\item $INV(X) \land T(X,Y,X') \Rightarrow INV(X')$
\item $INV(X) \Rightarrow P(X)$
\end{itemize}

To achieve this objective, it maintains a monotone CNF sequence $F_0, F_1 \ldots F_k$. Each $F_i$ in the sequence is called a \emph{frame}, which represents an over-approximation of the states reachable in $S$ within $i$ transition steps. Each clause $c$ in $F_i$ is called \emph{lemma} and the index of a frame is called \emph{level}. IC3 maintains the following invariant: 

\begin{itemize}
\item $F_0 = I$
\item for $i > 0, F_i$ is a set of clauses
\item $F_{i+1} \subseteq F_i$ ($F_i, F_{i+1}$ are sets of clauses)
\item $F_i \Rightarrow F_{i+1}$
\item $F_i \land T \Rightarrow F_{i+1}$
\item for $i < k, F_i \Rightarrow P$
\end{itemize}

A lemma $c$ is said to be \emph{inductive relative} to $F_i$ if, starting from the intersection of $F_i$ and $c$, all states reached in a single transition are located inside $c$. This condition can be expressed as a SAT query $sat(F_i \land c \land T \land \lnot c')$. If this query is satisfied, it indicates that $c$ is not inductive relative to $F_i$ because we can find a counterexample to induction (CTI) that starting from $F_i \land c$ and transitioning outside of the lemma $c$. If lemma $c$ is inductive relative to $F_i$, it can be also said that cube $\lnot c$ is blocked in $F_{i+1}$.

Algorithm \ref{alg:ic3} provides an overview of the IC3 algorithm. This algorithm incrementally constructs frames by iteratively performing two phases: the blocking phase and the propagation phase. During the blocking phase, the IC3 algorithm focuses on making $F_k \Rightarrow P$. It iteratively get a cube $c$ such that $c \models \lnot P$, and block it recursively. This process involves attempting to block the cube's predecessors if it cannot be blocked directly. It continues until the initial states cannot be blocked, indicating that $\lnot P$ can be reached from the initial states in $k$ transitions thus violating the property. In cases where a cube can be confirmed as blocked, IC3 proceeds to enlarge the set of blocked states through a process called generalization. This involves dropping variables and ensuring that the resulting clause remains relative inductive, with the objective of obtaining a minimal inductive clause. Subsequently, IC3 attempts to push the generalized lemma to the top and incorporates it into the frames. Throughout this phase, the sequence is augmented with additional lemmas, which serve to strengthen the approximation of the reachable state space. The propagation phase tries to push lemmas to the top. If a lemma $c$ in $F_i \setminus F_{i+1}$ is also inductive relative to $F_i$, then push it into $F_{i+1}$. During this process, if two consecutive frames become identical ($F_i = F_{i+1}$), then the inductive invariant is found and the safety of this model can be proofed.

\begin{algorithm}
\caption{Overview of IC3}
\label{alg:ic3}
\begin{algorithmic}[1]
\Function{$inductive\_relative$}{clause $c$, level $i$}
    \State \Return $\lnot sat(F_i\land c \land T \land \lnot c')$
\EndFunction
\\
\Function{$generalize$}{cube $b$, level $i$}
    \For{each $l \in b$}
        \State $cand \coloneqq b \setminus \{l\}$ \Comment{drop variable}
        \If{$I \Rightarrow \lnot cand$ and $inductive\_relative(\lnot cand, i-1)$}
            \State $b \coloneqq cand$
        \EndIf
    \EndFor
    \State \Return $b$
\EndFunction
\\
\Function{$block$}{cube $c$, level $i$, level $k$}
    \If{$i = 0$}
        \State \Return $false$
    \EndIf
    \While{$\lnot inductive\_relative(\lnot c, i-1)$}
        \State $p \coloneqq get\_predecessor(i-1)$
        \If{$\lnot block(p, i-1, k)$}
            \State \Return $false$
        \EndIf
    \EndWhile
    \State $mic \coloneqq \lnot generalize(c, i)$
    \While{$i<k$ and $inductive\_relative(mic, i)$}
        \State $i \coloneqq i+1$ \Comment{push lemma}
    \EndWhile
    \For{$1\leq j \leq i$}
        \State $F_j \coloneqq F_j \cup \{mic\}$
    \EndFor
    \State \Return $true$
\EndFunction
\\
\Function{$propagate$}{level $k$}
    \For{$1 \leq i < k$}
        \For{each $c \in F_i \setminus F_{i+1}$} 
            \If{$inductive\_relative(c, i)$} \Comment{push lemma}
                \State $F_{i+1} \coloneqq F_{i+1} \cup \{c\}$
            \EndIf
        \EndFor
        \If{$F_i = F_{i+1}$}
            \State \Return $true$
        \EndIf
    \EndFor
    \State \Return $false$
\EndFunction
\\
\Procedure{$ic3$}{$I,T,P$}
    \If{$sat(I\land\lnot P)$}
        \State \Return $unsafe$
    \EndIf
    \State $F_0 \coloneqq I,k \coloneqq 1,F_k \coloneqq \top$
    \While{$true$}
        \While{$sat(F_k\land \lnot P)$} \Comment{blocking phase}
            \State $c \coloneqq get\_predecessor(k)$
            \If{$\lnot block(c, k, k)$}
                \State \Return $unsafe$
            \EndIf
        \EndWhile
        \State $k \coloneqq k + 1,F_k \coloneqq \top$ \Comment{propagation phase}
        \If{$propagate(k)$}
            \State \Return $safe$
        \EndIf
    \EndWhile
\EndProcedure
\end{algorithmic}
\end{algorithm}

\section{Predicting Lemmas}
\label{Sec:PredictingLemmas}
\subsection{Intuition}
The generalization of the IC3 algorithm attempts to iteratively drop variables to obtain a minimal inductive clause. Each variable dropping corresponds to a SAT query, which can incur significant overhead. If there existed a method to speculatively predict a high-quality (minimal) inductive clause and demonstrate its relative inductiveness, it could improve the performance of generalization by avoiding dropping variables one by one.

In Figure \ref{fig:Fi1Fi}, when attempting to push the lemma $\lnot c_1$ and $\lnot c_2$ from $F_{i-1}$ to $F_i$ (where $c_1$ and $c_2$ are cubes), $c_1$ succeeds, but $c_2$ fails. This failure indicates that the cube $c_2$ cannot be blocked in $F_i$. However, it is highly likely that cube $c_2$ is generalized from state $b$, where $b$ is a bad state and $b \models c_2$. Since cube $c_2$ cannot be blocked in $F_i$, it is possible that $b$ still exists in $F_i$ ($b \models F_i$). IC3 will later find $b$ and attempt to block and generalize it again.

Additionally, in case of a push failure, we encounter two states in counterexample to propagation (CTP): $s$ and $t$, where $t$ is a successor state of $s$, $s \models F_{i-1}$, and $t \models c_2$. The existence of CTP prevents $c_2$ from being blocked in $F_i$. If we can remove state $t$ from $c_2$, CTP will no longer be effective to the remaining part, thus the remaining part may potentially be blocked in $F_i$. In this paper, we attempt to refine $c_2$ by utilizing state $t$ to derive a new cube $c_3$, where $t \not\models c_3$, $b \models c_3$ and $c_3$ can potentially be blocked in $F_i$, thus a possible lemma in $F_i$ is predicted.

% The method proposed in \cite{PushingToTheTop} aims to block $s$ to facilitate the push of $c_2$. However, the replication conducted in \cite{CAV23} indicate that this approach is not reproducible and can potentially lead to performance loss. 

\subsection{How to Refine Cubes}
How can we refine the cube $c_2$ using $t$ to obtain a cube $c_3$ such that $t \not\models c_3$ and $b \models c_3$? To achieve this objective, we introduce the concept of the \emph{diff set} of two cubes:

\begin{definition}[Diff Set]
\label{DefinitionDiffSet}
Let $a$ and $b$ are cubes, $diff(a, b) = \{l \mid l\in a \land \lnot l \in b$\}.
\end{definition}

Diff set $diff(a, b)$ represents the set of literals in $a$ such that their negations are in $b$. It should be noted that $diff(a, b) \neq diff(b, a)$. The diff set has the following theorems:

\begin{theorem}
\label{TheoremDiffSetEmpty}
Let $a$ and $b$ are cubes, $a \neq \bot$ and $b \neq \bot$. $a \land b = \bot$, iff $diff(a, b) \neq \emptyset$.
\end{theorem}

\begin{proof}
Let cube $c = a \land b$ ($c = a \cup b$). If $c = \bot$, it implies the presence of a literal $l$ in $c$ along with its negation, $\lnot l$. Since $a \neq \bot$ and $b \neq \bot$, we can conclude that there exists a literal in $a$ and its negation in $b$, indicating that $diff(a, b) \neq \emptyset$. Conversely, if $diff(a, b) \neq \emptyset$, we can identify a literal $l$ in $diff(a, b)$, where both $l$ and its negation exist in $c$, resulting in $a \land b = \bot$.
\end{proof}

\begin{theorem}
\label{TheoremTernaryDiffSet}
Let $a$, $b$, $c$ are cubes, if $diff(a, b) \neq \emptyset$ and $c \cap diff(a, b) \neq \emptyset$, then $diff(c, b) \neq \emptyset$.
\end{theorem}

\begin{proof}
Let $l$ be a literal, and suppose $l \in c \cap diff(a, b)$. It can be derived that $l \in c$. According to Definition \ref{DefinitionDiffSet}, we know that $l \in a$ and $\lnot l \in b$. Since $l \in c$ and $\lnot l \in b$, it follows that $l \in diff(c, b)$, leading to the conclusion that $diff(c, b) \neq \emptyset$.
\end{proof}

Additionally, cube has the following theorem:

\begin{theorem}
\label{TheoremCube}
Let $a$ and $b$ are cubes, $a \neq \bot$ and $b \neq \bot$. $a \Rightarrow b$ (or $a \models b$) iff $b \subseteq a$.
\end{theorem}

% \begin{proof}
% If $b \subseteq a$ and $a$ is true, all literals in $a$ must be true. And all literals in $b$ must be true as well, indicating that $a \Rightarrow b$. Thus, we can prove that if $b \subseteq a$, then $a \Rightarrow b$. To prove the converse, we assume $a \Rightarrow b$ and suppose that $b \not\subseteq a$ by contradiction. This implies the existence of a literal $l$ such that $l \in b$ and $l \notin a$. Let us consider the cube $c = a \cup \{\lnot l\}$. Since $a \subseteq c$, it follows that $c \Rightarrow a$ is true. Furthermore, from $a \Rightarrow b$, we can derive $c \Rightarrow b$. Since $l \in b$ and $\lnot l \in c$, we have $diff(c, b) \neq \emptyset$. This means that $c \land b = \bot$, and consequently, $\lnot c \lor \lnot b$ is true. Therefore, we can conclude that $c \Rightarrow \lnot b$, which contradicts the assumption of $c \Rightarrow b$. Thus, we have proven that if $a \Rightarrow b$, then $b \subseteq a$.
% \end{proof}

Actually in IC3, $t$ and $b$ both are cubes, that they may not only represent one state but a set of states. It is possible for cubes $t$ and $b$ to have intersections. However, we will initially consider the case where there is no intersection, i.e., $t \land b = \bot$. According to Theorem \ref{TheoremDiffSetEmpty}, we can derive the following equation:
\begin{equation}\label{DiffBT}
diff(b, t) \neq \emptyset
\end{equation}

Our objective is to find a cube $c_3$ that satisfies $c_3 \land t = \bot$, $b \models c_3$, and $c_3 \models c_2$. Referring to Theorem \ref{TheoremDiffSetEmpty} and Theorem \ref{TheoremCube}, we can achieve our objective if $c_3$ satisfies the following equations:
\begin{equation}\label{DiffC3T}
diff(c_3, t) \neq \emptyset
\end{equation}
\begin{equation}\label{C3b}
c_3 \subseteq b
\end{equation}
\begin{equation}\label{C2C3}
c_2 \subseteq c_3
\end{equation}
 
Building upon Theorem \ref{TheoremTernaryDiffSet} and Equation \ref{DiffBT}, we can establish the truth of Equation \ref{DiffC3T} by ensuring the satisfaction of the following equation:
\begin{equation}\label{C3DiffBT}
c_3 \cap diff(b, t) \neq \emptyset
\end{equation}

By constructing $c_3$ according to the following equation, we can satisfy Equation \ref{C3DiffBT}:
\begin{equation}\label{MakeC3}
c_3 = c_2 \cup \{l\}\ (l \in diff(b, t))
\end{equation}

The method we use for predicting lemmas is represented by Equation \ref{MakeC3}. This equation demonstrates the process of attempting to select an element from the diff set and incorporate it into $c2$, resulting in the creation of a new cube $c3$ that is larger than $c2$ by one element. If Equation \ref{MakeC3} is satisfied, it is evident that Equation \ref{C2C3} is true. Additionally, based on Definition \ref{DefinitionDiffSet}, Equation \ref{C3b} is also satisfied. As a result, we have successfully obtained the desired cube.

\subsection{The Algorithm}

\definecolor{deepblue}{RGB}{0,0,148}
\begin{algorithm}
\caption{Predicting Lemmas}
\label{alg:PredictingIc3}
\begin{algorithmic}[1]
\Function{$parent\_lemmas$}{clause $c$, level $i$}
    \State $parents \coloneqq \emptyset$
    \If{$i \neq 0$}
        \For{each $p \in F_i \setminus F_{i+1}$}
            \If{$p \subseteq c$}
                \State $parents \coloneqq parents \cup \{p\}$
            \EndIf
        \EndFor
    \EndIf
    \State \Return $parents$
\EndFunction
\\
\Function{$generalize$}{cube $b$, level $i$}
\textcolor{deepblue} {
    \State $parents \coloneqq parent\_lemmas(\lnot b, i-1)$ \label{BeginPredict}
    \For{each $p \in parents$}
        \If{$\lnot failure\_push.has(p, i-1)$} \Comment{find state $t$} \label{BeginFindFailed}
            \State \textbf{continue} \label{EndFindFailed}
        \EndIf
        \State $t \coloneqq failure\_push[(p, i-1)]$
        \State $ds \coloneqq diff(b, t)$ \Comment{get diff set}
        \If{$ds = \emptyset$} \Comment{push parent lemma} \label{BeginDiffEmpty}
            \If{$inductive\_relative(p, i-1)$}
                \State \Return $p$
            \Else
                \State $failure\_push[(p, i-1)] \coloneqq get\_model(i-1)$
            \EndIf \label{EndDiffEmpty}
        \Else
            \For{each $d \in ds$} \label{BeginCand}
                \State $cand \coloneqq p \cup \{\lnot d\}$
                \If{$inductive\_relative(cand, i-1)$}
                    \State \Return $\lnot cand$
                \Else
                    \State $ds \coloneqq ds \cap diff(b, get\_model(i-1))$
                \EndIf \label{EndCand}
            \EndFor \label{EndPredict}
        \EndIf
    \EndFor
}
    \For{each $l \in b$}
        \State $cand \coloneqq b \setminus \{l\}$ \Comment{drop variable}
        \If{$I \Rightarrow \lnot cand$ and $inductive\_relative(\lnot cand, i-1)$}
            \State $b \coloneqq cand$
        \EndIf
    \EndFor
    \State \Return $b$
\EndFunction
\\
\Function{$block$}{cube $c$, level $i$, level $k$}
    \State \ldots
    \While{$i<k$ and $inductive\_relative(mic, i)$} \label{BlockPushLemma}
        \State $i \coloneqq i+1$ \Comment{push lemma}
    \EndWhile
        \State \textcolor{deepblue}{$failure\_push[(mic, i)] \coloneqq get\_model(i)$} \Comment{\textcolor{deepblue}{record state $t$}}
    \For{$1\leq j \leq i$}
        \State $F_j \coloneqq F_j \cup \{mic\}$
    \EndFor
    \State \Return $true$
\EndFunction
\\
\Function{$propagate$}{level $k$}
    \State \textcolor{deepblue}{$failure\_push.clear()$} \Comment{\textcolor{deepblue}{reconstruct hash table}}
    \For{$1 \leq i < k$}
        \For{each $c \in F_i \setminus F_{i+1}$}
            \If{$inductive\_relative(c, i)$} \Comment{push lemma} \label{PropagatePushLemma}
                \State $F_{i+1} \coloneqq F_{i+1} \cup \{c\}$
\textcolor{deepblue} {
            \Else \Comment{record state $t$}
                \State $failure\_push[(c, i)] \coloneqq get\_model(i)$
}
            \EndIf
        \EndFor
        \If{$F_i = F_{i+1}$}
            \State \Return $true$
        \EndIf
    \EndFor
    \State \Return $false$
\EndFunction
\end{algorithmic}
\end{algorithm}

Algorithm \ref{alg:PredictingIc3} introduces our proposed approach for lemma prediction by leveraging the CTP in IC3. The modifications to the original IC3 algorithm are denoted by the blue-colored code. In line \ref{BlockPushLemma} and line \ref{PropagatePushLemma}, the algorithm attempts to push lemmas. If push fails, it retrieves the corresponding state $t$ from the SAT solver and stores it in a hash table. The hash table uses a tuple consisting of the failed lemma and the current level as the key. Additionally, at regular intervals before each propagation step, the hash table is cleared and reconstructed.

In the generalization, the algorithm first attempts to predict a possible minimal lemma and checks its validity, as demonstrated in lines \ref{BeginPredict} to \ref{EndPredict}. Initially, it identifies all the parent lemmas of clause $\lnot b$ in level $i$. A \emph{parent lemma} of clause $c$ at level $i$ is the lemma $p$ only in the frame of level $i-1$ (not in the current level) that $p \Rightarrow c$. For example, it can be observed that $\lnot c_2$ is the parent lemma of clause $\lnot b$ at level $i$ in Figure \ref{fig:Fi1Fi}. We make an attempt to predict a minimal lemma in $F_i$ based on each identified parent lemma:
\begin{enumerate}[leftmargin=*]
\item If the parent lemma $p$ does not have any failed push attempts at level $i-1$, we cannot find CTP, and thus cannot predict a lemma and proceed with the next parent lemma as shown in lines \ref{BeginFindFailed} - \ref{EndFindFailed}.
\item Upon successfully retrieving state $t$ from the hash table, we proceed to calculate the diff set $ds$ of the cube $b$ and cube $t$. In lines \ref{BeginDiffEmpty} - \ref{EndDiffEmpty}, if the diff set is empty, it indicates an intersection between the cubes $b$ and $t$ according to Theorem \ref{TheoremDiffSetEmpty}, meaning that some states in cube $t$ are already blocked due to the blocking of $b$ in $F_i$. In this case, we consider pushing the parent lemma $p$ to $F_i$ as a predicted lemma. If successful, the predicted lemma is deemed valid, we take it as the final result. Otherwise, we store the state $t$ associated with this failure in the hash table.
\item If the diff set is not empty, we attempt to predict a possible lemma using Equation \ref{MakeC3} by iteratively selecting a literal from the diff set as shown in lines \ref{BeginCand} - \ref{EndCand}. If successful, since it is only one variable longer than the parent lemma, we consider it to be a high-quality lemma. Therefore, we no longer generalize it by dropping variables and directly treat it as the final result of generalization. If unsuccessful, it is highly likely that the counterexample is also an another CTP of pushing $p$ to $F_i$, thus we recalculate the diff set and eliminate certain infeasible candidates from it.
\end{enumerate}

\section{Evaluation}
\label{Sec:Evaluation}
\subsection{Experimental Setup}
\emph{IC3ref} \cite{IC3ref} is an IC3 implementation provided by the inventor of this algorithm, which is competitive and commonly used as the baseline \cite{FBPDR, IC3Progress, CAV23}. We also implemented IC3 algorithm in Rust, a modern programming language designed to offer both performance and safety, denoted as \emph{RIC3}, which is competitive to IC3ref. We integrated our proposed predicting lemmas optimization into both IC3ref and RIC3, denoted as \emph{IC3ref-pl} and \emph{RIC3-pl} respectively. Since our method is related to \cite{CAV23}, which is also implemented in IC3ref, available at \cite{CAV23Repo}, we also conducted an evaluation of it, denoted as \emph{IC3ref-CAV23}. We also consider the PDR implementation in the ABC \cite{ABC}, which is a competitive model checker.

We conducted all experiments using the complete benchmark set of HWMCC'15 (excluding restricted-access cases) and HWMCC'17, totaling 730 cases, under consistent resource constraints (8 GB, 1000s). The experiments were performed on an AMD EPYC 7532 machine with 32 cores running at 2.4 GHz. To ensure reproducibility, we have provided our experiment implementations \cite{Artifact}.

\subsection{Experimental Results}
The summary of results, comparisons among different configurations, and scatter plots of implementations with and without lemma prediction are presented in Table \ref{tab:OverallResult}, Figure \ref{fig:PlotResult}, and Figure \ref{fig:Scatters}, respectively. It can be observed that by incorporating our proposed optimization, different implementations solve more cases compared to the original implementation within various time limits and a greater number of cases have been solved faster. It can also be observed that our proposed method solves more cases than IC3ref-CAV23 and ABC-PDR. This demonstrates the effectiveness of our proposed method in enhancing the performance of IC3.

\begin{table}
\caption{Summary of Results}
\label{tab:OverallResult}
\begin{tabular}{c c c c c}
\hline
Configuration & Solved & Safe & Unsafe\\
\hline
RIC3 & 365 & 264 & 101 \\ 
RIC3-pl & 375 & 273 & 102 \\
IC3ref & 371 & 263 & 108 \\
IC3ref-pl & 379 & 268 & 111 \\
IC3ref-CAV23 & 375 & 269 & 106 \\
ABC-PDR & 373 & 267 & 106 \\
\hline
\end{tabular}
\end{table}

\begin{figure}
    \centering
    \includegraphics[width=245pt]{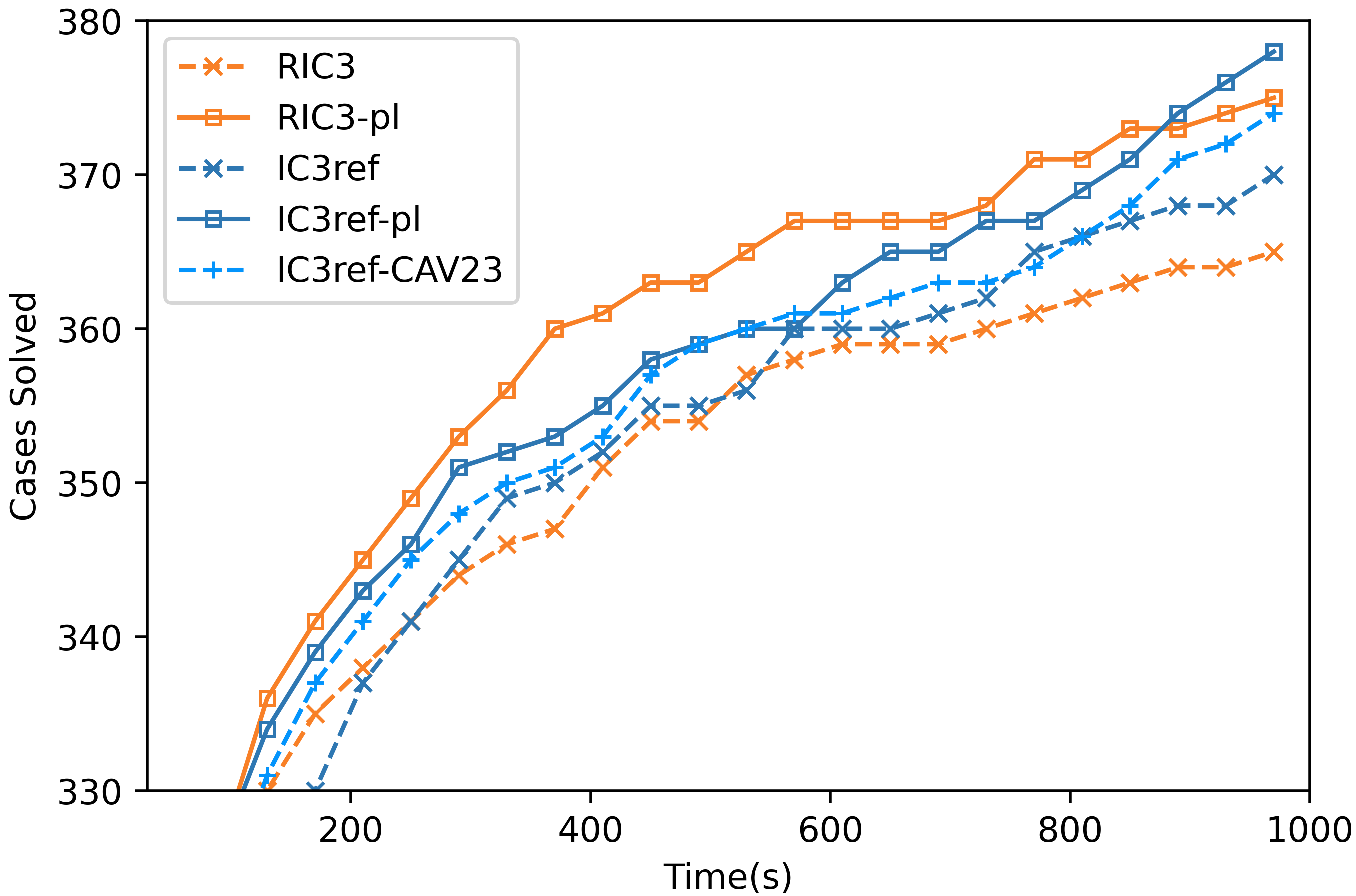}
    \caption{Comparisons among the different configurations.}
    \label{fig:PlotResult}
\end{figure}

\begin{figure}
    \centering
    \includegraphics[width=245pt]{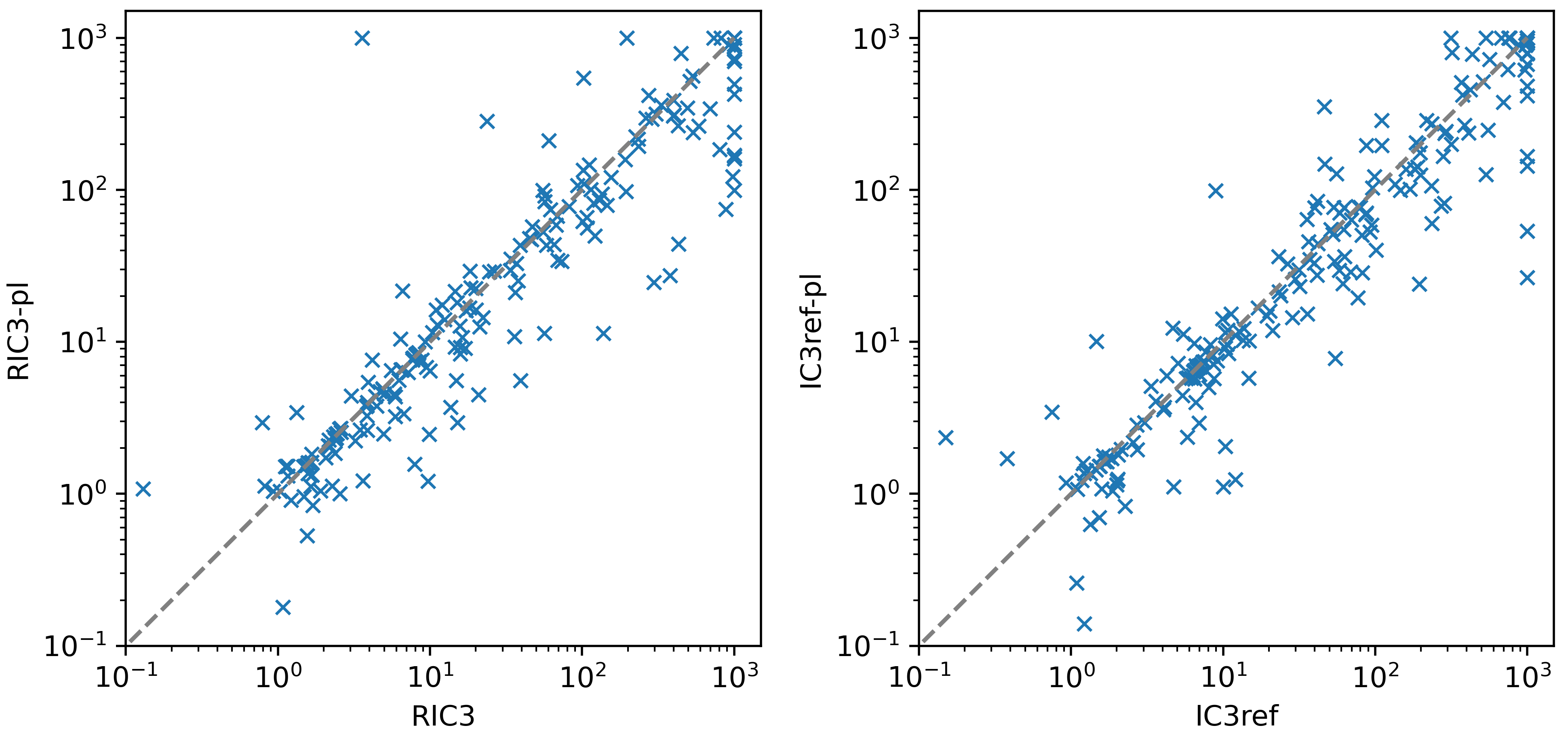}
    \caption{Scatters of RIC3 and IC3ref  with and without the proposed optimization. Points below the diagonal indicate better performance with the proposed optimization active.}
    \label{fig:Scatters}
\end{figure}

\subsection{Analysis}
\begin{table}
\caption{Average Success Rates}
\label{tab:SuccessRate}
\begin{tabular}{c c c c}
\hline
Configuration & Avg $SR_{lp}$ & Avg $SR_{fp}$ & Avg $SR_{adv}$\\
\hline
RIC3-pl & 38.61\% & 40.67\% & 24.03\% \\
IC3ref-pl & 31.5\% & 37.81\% & 19.46\% \\
\hline
\end{tabular}
\end{table}

\begin{figure}
    \centering
    \begin{subfigure}{0.45\textwidth}
        \centering
        \includegraphics[width=245pt]{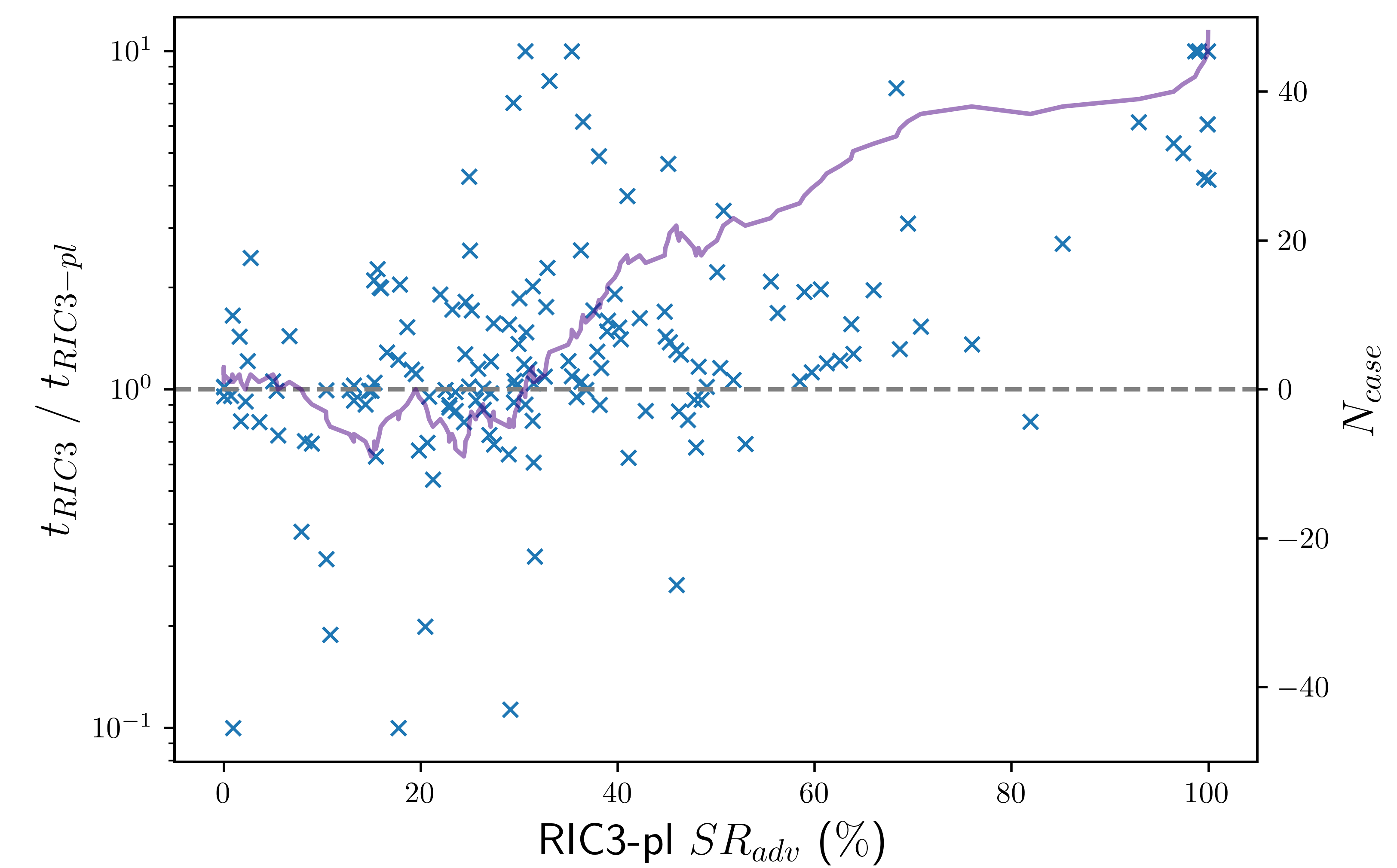}
    \end{subfigure}
    \vskip 0cm
    \begin{subfigure}{0.45\textwidth}
        \centering
        \includegraphics[width=245pt]{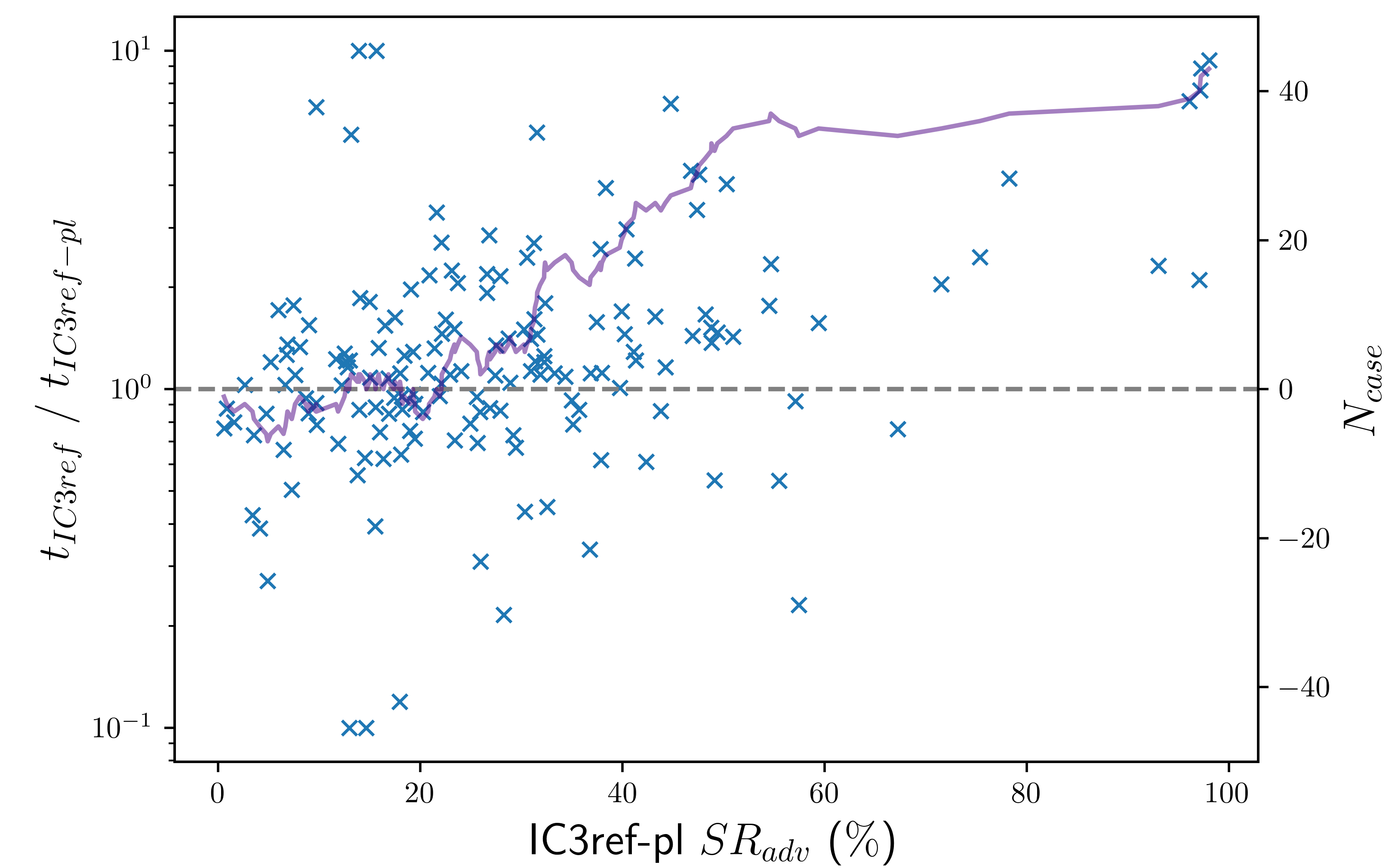}
    \end{subfigure}

    \caption{Comparing the ratio of runtime of different implementations without and with our proposed optimization (left y-axis) with the success rate of avoiding dropped variables $SR_{adv}$ during generalization (x-axis). The right y-axis represents the cumulative number of cases with improved performance as the prediction accuracy increases. Cases with runtime both timeout or both less than 1s are ignored.}
    \label{fig:SuccessRate}
\end{figure}

To better demonstrate the effectiveness of our proposed method, we introduced several success rates: the lemma prediction success rate ($SR_{lp} = N_{sp} / N_{p}$), the success rate of finding failed pushed parent lemmas ($SR_{fp} = N_{fp} / N_g$), and the success rate of avoiding dropping variables ($SR_{adv} = N_{sp} / N_g$). Here, $N_{sp}$ represents the number of successful lemma predictions, $N_p$ represents the total number of lemma predictions (the number of SAT queries in prediction), $N_{fp}$ represents the number of successfully finding failed pushed parent lemmas, and $N_g$ represents the total number of generalizations. Table \ref{tab:SuccessRate} presents the three kind of average success rates of different implementations with our proposed optimization. It can be observed that our method can achieve a considerable success rate if the failed pushed lemma is found. Figure \ref{fig:SuccessRate} depicts the correlation between $SR_{adv}$ and the checking time. It can be observed that higher prediction success rates result in a better performance in the proportion of cases. This observation supports the conjecture that higher accuracy in predictions can yield better performance.

\section{Related Work}
\label{Sec:RelatedWork}
After the introduction of IC3 \cite{IC3}, numerous works have been proposed to improve the efficiency. \cite{FBPDR} interleaves a forward and a backward execution of IC3 and strengthens one frame sequence by leveraging the other. \cite{IC3Progress} introduces under-approximation to improve the performance of bug-finding.

When a clause fails to be relatively inductive, a common approach is to identify the counterexample of this failure and attempting to render it ineffective. \cite{BetterGeneralization} attempts to block counterexample to generalization (CTG) to make the variable dropped clause inductive. \cite{PushingToTheTop} aggressively pushes lemmas to the top by adding may-proof-obligation (blocking counterexamples of push failure). The key distinction between our proposed method and these approaches lies in the fact that while they all focus on blocking the predecessor state (state $s$ in Figure \ref{fig:Fi1Fi}) of a counterexample, we utilize the successor state (state $t$) to predict a possible lemma.

Another related work is \cite{CAV23}, which also attempts to find the parent lemma and doesn't dropping literals which existed in the parent lemma during generalization. Our approach differs significantly as we attempt to avoid dropping variables by refining failed pushed parent lemma as a prediction.

\section{Conclusion and Future Work}
\label{Sec:Conclusion}
In this paper, we proposed a novel method for predicting minimal lemmas before dropping variables by utilizing the counterexample to propagation (CTP). By making successful predictions, we can effectively circumvent the need to drop variables during generalization, which is a significant overhead in the IC3 algorithm. The experimental results demonstrate a commendable success rate in lemma prediction, consequently leading to a significant performance improvement. In the future, we plan to improve the prediction rate of our proposed method in order to avoiding dropping variables as much as possible.

\begin{acks}
This work was supported by the Basic Research Projects from the Institute of Software, Chinese Academy of Sciences (Grant No. ISCAS-JCZD-202307) and the National Natural Science Foundation of China (Grant No. 62372438).
\end{acks}

\bibliographystyle{ACM-Reference-Format}
\bibliography{pred}

\end{document}